\documentclass{article}
\usepackage[left=2.5cm, top=3cm, bottom=3cm, right=2.5cm]{geometry}
\usepackage{amsfonts}
\usepackage{amsthm}
\usepackage{amsmath}
\usepackage{enumerate}
\usepackage[numbers,sort&compress]{natbib}

\newtheorem{theorem}{Theorem}[section]
\newtheorem{lemma}{Lemma}[section]

\theoremstyle{definition}

\newcommand{\secref}[1]{Section~\ref{#1}}

\newcommand{\thmref}[1]{Theorem~\ref{#1}}

\newcommand{\itref}[1]{(\ref{#1})}

\newcommand{\lshad}{[\![}
\newcommand{\rshad}{]\!]}
\newcommand{\sdot}{\,\cdot\,}

\newcommand{\nbr}[1]{$#1$\nobreakdash-\hspace{0pt}}
\newcommand{\Pol}{\mathsf{Pol}}

\DeclareMathOperator{\id}{id}

\everymath{\displaystyle}
\numberwithin{equation}{section}

\title{On local equivalence of star-products on Poisson manifolds}
\author{Ziemowit Doma{\'n}ski\\
\small Institute of Mathematics, Pozna{\'n} University of Technology\\
\small Piotrowo 3A, 60-965 Pozna{\'n}, Poland\\
\small \tt ziemowit.domanski@put.poznan.pl \and Maciej B{\l}aszak\\
\small Faculty of Physics, Division of Mathematical Physics,
       Adam Mickiewicz University\\
\small Umultowska 85, 61-614 Pozna{\'n}, Poland\\
\small \tt blaszakm@amu.edu.pl}

\begin{document}

\maketitle

\begin{abstract}
We present a proof that every star-product defined on a Poisson manifold
and written in a given quantum canonical coordinate system is uniquely
equivalent with a Moyal product associated with this coordinate system.
The equivalence is assumed to satisfy some additional conditions which guarantee
its uniqueness. Moreover, the systematic construction of such equivalence is
presented and a formula for this equivalence in a case of a particular class of
star-products is given, to the fourth order in $\hbar$.
\\[\baselineskip]
\end{abstract}

\section{Introduction}
\label{sec:2}
One of the admissible methods of quantization of a classical Hamiltonian system
is a deformation quantization procedure. In this procedure one deforms a
classical Poisson algebra $C^\infty(M)$ of smooth complex-valued functions
defined on a phase space $M$ (Poisson manifold) to an appropriate noncommutative
algebra \cite{Moyal:1949,Bayen:1978a,Bayen:1978b}. The noncommutative product
in this algebra is usually denoted by $\star$ and called a star-product.

The existence of a star-product on any symplectic manifold was first proved in
1983 by \citet{DeWilde:1983}. Later \citet{Fedosov:1994} gave a recursive
construction of a star-product on a symplectic manifold using the framework of
Weyl bundles. Independently, \citet{Omori:1991} gave an alternative proof of
the existence of a star-product on a symplectic manifold, also using the
framework of Weyl bundles. Finally, in 1997, \citet{Kontsevich:2003} proved the
existence of a star-product on any Poisson manifold.

Two star-products on a given Poisson manifold may not be equivalent. In fact
the equivalence classes of star-products on a symplectic manifold $M$ are
parametrized by formal series of elements in the second de~Rham cohomology space
of $M$, $H^2(M;\mathbb{C})\lshad\hbar\rshad$ \cite{Nest:1995b,Bertelson:1997,%
Deligne:1995}.

If we choose some coordinate system on a domain $\mathcal{O} \subset M$ of the
Poisson manifold $M$, then a given star-product can be written locally in this
coordinate system. The simplest case is when a coordinate representation of the
star-product is in the form of a Moyal star-product \cite{Moyal:1949}.

To each star-product corresponds a distinct class of coordinate systems, namely
quantum canonical coordinate systems. Coordinates which are canonical with
respect to one star-product do not have to be canonical with respect to the
other star-product. If $\mathcal{O} \subset M$ is a domain of some coordinate
system $\varphi \colon \mathcal{O} \to \mathbb{R}^{2N}$ then equivalence classes
of star-products written in these coordinates are parametrized by elements of
$H^2(\mathcal{O};\mathbb{C})\lshad\hbar\rshad$. The Moyal star-product is in one
of these classes. Let us denote this class by
$\mathcal{S}(\mathcal{O},\varphi)$. So every star-product on $M$ which
coordinate representation with respect to the coordinate chart
$(\mathcal{O},\varphi)$ is in the class $\mathcal{S}(\mathcal{O},\varphi)$ is
locally equivalent with the Moyal star-product. For part of these star-products,
the coordinates $(\mathcal{O},\varphi)$ are quantum canonical, like for the
Moyal product. The class of such star-products we will denote by
$\mathcal{S}_q(\mathcal{O},\varphi)$.

The star-products in $\mathcal{S}_q(\mathcal{O},\varphi)$ can be used to perform
nonequivalent quantizations of a classical Hamiltonian system. For this reason
the knowledge of morphisms relating the star-products in
$\mathcal{S}_q(\mathcal{O},\varphi)$ with the Moyal product can help in
establishing the relations between received nonequivalent quantizations.
Moreover, the fact that these star-products are equivalent with the Moyal
product is useful when constructing particular realizations of quantizations.
This is because we know how to perform quantization by means of the Moyal
product. Especially, we know how to construct an operator representation in the
Hilbert space of a quantized classical system. Thus, if for a quantum canonical
coordinate system the star-product is equivalent with a Moyal product then the
problem can be reduced to the Moyal case
\cite{Blaszak:2012,Blaszak:2013,Blaszak:2013b}. For this reason we need a form
of an equivalence morphism $S$, which gives the equivalence with the Moyal
product. In this paper we show that every star-product on a general Poisson
manifold, written in quantum canonical coordinate system on
$\mathcal{O}\subset M$, is locally (globally in particular) equivalent with the
Moyal product in these coordinates. Moreover, we present a systematic
construction of the corresponding equivalence morphism $S$, order by order in
$\hbar$ (\secref{sec:1}). In particular, we derive the form of the morphism $S$
(to the fourth order in $\hbar$) for a class of star-products on a phase space
in the form of a cotangent bundle $T^*\mathcal{Q}$ to some manifold
$\mathcal{Q}$ (configuration space), where the star-products are generated by
flat connections on $\mathcal{Q}$. It should be noted, that in this case the
\nbr{\hbar}expansion of $S$ has only finite number of terms which will give
non-zero contribution when acting on functions polynomial in momenta. For
instance, to calculate the action of $S$ on functions at most cubic in momenta
we only need $S$ to the second order in $\hbar$, and for functions at most of
fifth order in momenta the expansion of $S$ to the fourth order in $\hbar$ is
required. Further on, we consider a class of star-products on a general
symplectic manifold, generated by symplectic connections, and construct the
corresponding morphism $S$ to the second order in $\hbar$ (\secref{sec:3}).

\section{Construction of the equivalence}
\label{sec:1}
Let $(M,\mathcal{P})$ be a general Poisson manifold, where $\mathcal{P}$ is its
Poisson tensor, and $\{\sdot,\sdot\}$ a Poisson bracket associated to
$\mathcal{P}$. Denote by $\mathbb{C}\lshad\hbar\rshad$ the ring of formal power
series in the parameter $\hbar$ with coefficients in $\mathbb{C}$ and by
$C^\infty(M)\lshad\hbar\rshad$ the space of formal power series in $\hbar$ with
coefficients in $C^\infty(M)$. On the Poisson manifold $(M,\mathcal{P})$ we
define a star-product as a bilinear map
\begin{equation}
C^\infty(M) \times C^\infty(M) \to C^\infty(M)\lshad\hbar\rshad, \quad
(f,g) \mapsto f \star g = \sum_{k=0}^\infty \hbar^k C_k(f,g),
\label{eq:1.1}
\end{equation}
which extends \nbr{\mathbb{C}\lshad\hbar\rshad}linearly to
$C^\infty(M)\lshad\hbar\rshad \times C^\infty(M)\lshad\hbar\rshad$, such that
\begin{enumerate}[(i)]
\item $C_k$ are bidifferential operators,
\label{item:1.1}

\item $C_0(f,g) = fg$,
\label{item:1.2}

\item $C_1(f,g) - C_1(g,f) = i\{f,g\}$,
\label{item:1.3}

\item $\sum_{l=0}^k \bigl(C_l(C_{k-l}(f,g),h) - C_l(f,C_{k-l}(g,h))\bigr) = 0$,
\label{item:1.4}

\item $C_k(f,1) = C_k(1,f) = 0$ for $k \ge 1$.
\label{item:1.5}
\end{enumerate}
Moreover, we define a deformed Poisson bracket by the formula
\begin{equation}
\lshad f,g \rshad_\star = \frac{1}{i\hbar} [f,g]
= \frac{1}{i\hbar}(f \star g - g \star f).
\end{equation}
The \nbr{\star}product and the deformed Poisson bracket have the following
properties:
\begin{enumerate}[(a)]
\item $f \star g = fg + o(\hbar)$,
\label{item:1.6}

\item $\lshad f,g \rshad = \{f,g\} + o(\hbar)$,
\label{item:1.7}

\item $f \star (g \star h) = (f \star g) \star h$ (associativity),
\label{item:1.8}

\item $f \star 1 = 1 \star f = f$.
\label{item:1.9}
\end{enumerate}
Properties~\itref{item:1.6} and~\itref{item:1.7} follows respectively from
\itref{item:1.2} and \itref{item:1.3}, property~\itref{item:1.8} is a result of
\itref{item:1.4}, and property~\itref{item:1.9} follows from \itref{item:1.5}.
Thus the space $C^\infty(M)\lshad\hbar\rshad$ endowed with the
\nbr{\star}product and the deformed Poisson bracket $\lshad\sdot,\sdot\rshad$
is a deformation of the classical Poisson algebra $C^\infty(M)$.

On a Poisson manifold $\mathbb{R}^d$, where $d = 2n + k$, endowed with a
canonical Poisson tensor
\begin{equation}
(\mathcal{P}^{\mu\nu}) = \begin{pmatrix}
0_n & I_n & 0_k \\
-I_n & 0_n & 0_k \\
0_k & 0_k & 0_k
\end{pmatrix}
\label{eq:1.3}
\end{equation}
the simplest star-product is a Moyal product
\begin{equation}
f \star_M g = f \exp\left(\frac{1}{2}i\hbar\mathcal{P}^{\mu\nu}
    \overleftarrow{\partial}_{x^\mu} \overrightarrow{\partial}_{x^\nu}\right) g.
\label{eq:1.4}
\end{equation}

Another example of the star-product \eqref{eq:1.1} used in a quantization
procedure is a product of the form
\begin{align}
f \star g & = f \exp\left(\frac{1}{2}i\hbar\mathcal{P}^{\mu\nu}
    \overleftarrow{D}_\mu \overrightarrow{D}_\nu \right) g \nonumber \\
& = \sum_{k=0}^\infty \frac{1}{k!} \left(\frac{i\hbar}{2}\right)^k
    \mathcal{P}^{\mu_1 \nu_1} \dotsm \mathcal{P}^{\mu_k \nu_k}
    (D_{\mu_1} \dotsm D_{\mu_k}f)(D_{\nu_1} \dotsm D_{\nu_k}g),
\label{eq:1.5}
\end{align}
where $\mathcal{P}^{\mu\nu}$ are given by \eqref{eq:1.3} and $D_1,\dotsc,D_d$
are globally defined pair-wise commuting vector fields such that
\begin{equation}
\mathcal{P} = \mathcal{P}^{\mu\nu} D_\mu \otimes D_\nu.
\label{eq:1.6}
\end{equation}
The star-product \eqref{eq:1.5} is well defined on Poisson manifolds
$(M,\mathcal{P})$ for which a Poisson tensor $\mathcal{P}$ can be globally
written in the form \eqref{eq:1.6}.

The sequence of vector fields $D_1,\dotsc,D_d$ is not uniquely
specified by the condition \eqref{eq:1.6} but there exists the whole family of
such sequences. Every such sequence of vector fields defines a star-product of
the form \eqref{eq:1.5}. Thus there exists the whole family of star-products
\eqref{eq:1.5} associated to the same Poisson tensor $\mathcal{P}$. All these
star-products are related to each other by automorphisms of the Poisson
manifold $(M,\mathcal{P})$, i.e. if $\star$ and $\star'$ are star-product
\eqref{eq:1.5} induced by sequences of vector fields $D_1,\dotsc,D_d$ and
$D'_1,\dotsc,D'_d$ then there exists an automorphism $T$ such that
\begin{equation}
(D_\mu f) \circ T = D'_\mu(f \circ T),
\end{equation}
from which follows that
\begin{equation}
(f \star g) \circ T = (f \circ T) \star' (g \circ T).
\end{equation}
If, in particular, we choose some global classical canonical coordinate system
$(x^1,\dotsc,x^d)$ on $M$ then a star-product \eqref{eq:1.5} induced by
coordinate vector fields $\partial_{x^1},\dotsc,\partial_{x^d}$ is a Moyal
product in these coordinates and all other star-products from the family
\eqref{eq:1.5} are related to the Moyal product by a classical canonical
coordinate transformation $T$. More details of the use of the star-product
\eqref{eq:1.5} in a quantization procedure the reader can find in
\cite{Blaszak:2013b}.

Recall that a coordinate system $(x^1,\dotsc,x^d)$ is classical canonical iff
\begin{equation}
\{x^\mu,x^\nu\} = \mathcal{P}^{\mu\nu},
\end{equation}
where $\mathcal{P}^{\mu\nu}$ are given by \eqref{eq:1.3}. In a complete analogy
a coordinate system $(x^1,\dotsc,x^d)$ is called quantum canonical iff
\begin{equation}
\lshad x^\mu,x^\nu \rshad = \mathcal{P}^{\mu\nu}.
\label{eq:1.7}
\end{equation}

As a motivation for further considerations let us consider a classical
Hamiltonian system described by a phase space (Poisson manifold)
$(M,\mathcal{P})$. The quantization of this system in accordance to deformation
quantization theory is performed by introducing a star-product $\star$ on $M$
and associating to each measurable quantity smooth complex-valued function
defined on $M$. Then, equivalently, the same quantization could be described by
another star-product $\star'$ equivalent with the star-product $\star$ and by
such assignment of functions to measurable quantities that to a given measurable
quantity corresponds function $S^{-1}f$ where $f$ is a function from the first
quantization scheme corresponding to the same measurable quantity, and $S$ is
the equivalence morphism between star-products $\star'$ and $\star$.

Quite often we are interested in a local description of a quantization. That is,
if $(x^1,\dotsc,x^d)$ are coordinates on a Poisson manifold $M$, defined on a
domain $\mathcal{O} \subset M$ and whose image is an open subset
$U \subset \mathbb{R}^d$, then we can write the \nbr{\star}product
in these coordinates receiving a product in the algebra
$C^\infty(U)\lshad\hbar\rshad$ denoted hereafter by $\star^{(x)}$. If, moreover,
the coordinates $(x^1,\dotsc,x^d)$ are at the same time classical and quantum
canonical, then in these coordinates the components of the Poisson tensor
$\mathcal{P}$ take the form \eqref{eq:1.3}. On $U$ we can also define a Moyal
product \eqref{eq:1.4} associated to the same Poisson tensor $\mathcal{P}$.
In what follows we will prove that in this case \nbr{\star^{(x)}}product and
the Moyal product are equivalent. Thus, locally the quantization given by
the \nbr{\star}product is equivalent with the Moyal quantization, provided that
to measurable quantities we will assign functions of the form $S^{-1}f$ as
explained above. In other words, locally a given quantization can be described
in terms of the Moyal product, however the exact form of the equivalence
morphism $S$ between the Moyal product and the \nbr{\star}product is needed.

As an example let us consider a Poisson manifold $M$ in the form of a cotangent
bundle $T^*\mathcal{Q}$ to a Riemannian manifold $\mathcal{Q}$. Assume we
perform a quantization, of a classical system described by such Poisson
manifold, by means of some \nbr{\star}product. We may want to construct an
operator representation, of the received quantum system, in the Hilbert space
$L^2(\mathcal{Q})$. That is, if we choose some coordinate system
$(q^1,\dotsc,q^N)$ on $\mathcal{Q}$, then the induced canonical coordinate
system $(q^1,\dotsc,q^N,p_1,\dotsc,p_N)$ on $T^*\mathcal{Q}$ will be quantum
canonical. We may want then to prescribe to functions $f$ on $T^*\mathcal{Q}$
written in the canonical coordinates $(q^1,\dotsc,q^N,p_1,\dotsc,p_N)$
appropriately ordered operator functions $f(\hat{q},\hat{p})$, where
$\hat{q}^i,\hat{p}_j$ are operators of position and momentum corresponding to
the coordinate system $(q^1,\dotsc,q^N,p_1,\dotsc,p_N)$.

The appropriate operator representation may be constructed using the property
that the given quantization is locally equivalent with the Moyal quantization.
Since we know that for the Moyal quantization corresponds Weyl (symmetric)
ordering of operators $\hat{q}^i,\hat{p}_j$ then to a function $f$ on
$T^*\mathcal{Q}$ written in the canonical coordinates
$(q^1,\dotsc,q^N,p_1,\dotsc,p_N)$ should correspond Weyl ordered operator
function $S^{-1}f(\hat{q},\hat{p})$, where $S$ is the equivalence morphism
between the Moyal product and the \nbr{\star}product. The formula
$S^{-1}f(\hat{q},\hat{p})$ could be viewed as a definition of a new
\nbr{S}ordering of operators $\hat{q}^i,\hat{p}_j$. Note, that for a different
coordinate system $(q'^1,\dotsc,q'^N,p'_1,\dotsc,p'_N)$ we would get different
morphism $S$ and different \nbr{S}ordering, so that the operator representation
will be consistent with the change of coordinates. Note also, that for natural
star-products on $T^*\mathcal{Q}$ introduced in Section~\ref{sec:3} the
corresponding equivalence morphism $S$ with the Moyal product will have such
property that in its \nbr{\hbar}expansion only finite number of terms will give
non-zero contribution when acting on functions polynomial in momenta. For
instance, if $f(q,p) = K^{ij}(q) p_i p_j$, then the action of the morphism $S$
given by \eqref{eq:3.37} results in the following function
\begin{equation}
S^{-1}f(q,p) = K^{ij}(q) p_i p_j
    - \frac{\hbar^2}{4} \left(K^{ij}_{\phantom{ij},k}(q) \Gamma^k_{ij}(q)
    + K^{ij}(q) \Gamma^k_{li}(q) \Gamma^l_{kj}(q) \right).
\end{equation}
More details on the presented approach to quantization and the construction of
the operator representation the reader can find in
\cite{Blaszak:2012,Blaszak:2013,Blaszak:2013b}.

The above example shows that it is important to have a systematic construction
of the equivalence morphism $S$. In the following theorem formulas for the
construction of the morphism $S$ order by order in $\hbar$ are given, together
with the proof of the existence of the equivalence morphism $S$ for every
star-product. The proof of the existence of the morphism $S$ is based on the
results of \cite{Dito:1999}.

\begin{theorem}
\label{thm:1.1}
For a Poisson manifold $(M,\mathcal{P})$ together with a star-product $\star$
defined on it and for any coordinate system $(x^1,\dotsc,x^d)$ on $M$, whose
image is an open subset $U \subset \mathbb{R}^d$, and which is at the same time
classical and quantum canonical with respect to the \nbr{\star}product, there
exists a unique series $S$ of the form
\begin{equation}
S = \id + \sum_{k=1}^\infty \hbar^k S_k,
\end{equation}
where $S_k$ are differential operators on $C^\infty(U)\lshad\hbar\rshad$,
such that
\begin{subequations}
\label{eq:1.8}
\begin{align}
S(f \star_M^{(x)} g) & = Sf \star^{(x)} Sg, \label{eq:1.8a} \\
Sx^\alpha & = x^\alpha, \label{eq:1.8b}
\end{align}
\end{subequations}
where $\star_M^{(x)}$ is a star-product which in the coordinates
$(x^1,\dotsc,x^d)$ is of the form of the Moyal product. The operators $S_k$
will satisfy the following recurrence relations for $k \ge 1$
\begin{equation}
[S_k,x^\alpha](f) = \frac{1}{2} \sum_{l=1}^k \bigl(C_l(x^\alpha,S_{k-l}(f))
    + C_l(S_{k-l}(f),x^\alpha)\bigr), \quad f \in C^\infty(M).
\label{eq:1.9}
\end{equation}
\end{theorem}

\begin{proof}
We will show that the searched morphism $S$ on monomials takes the form
\begin{equation}
S(x^{\alpha_1} \dotsm x^{\alpha_r}) =
    \frac{1}{r!} \sum_{\sigma \in \mathfrak{S}_r}
    x^{\sigma(\alpha_1)} \star^{(x)} \dotsm \star^{(x)} x^{\sigma(\alpha_r)},
\label{eq:1.10}
\end{equation}
where $\mathfrak{S}_r$ is the group of all permutations of the set
$\{1,2,\dotsc,r\}$. The morphism $S$ can be then linearly extended to the space
of all polynomials $\Pol$. We will prove that $S$ can be uniquely extended to
the space $C^\infty(U)\lshad\hbar\rshad$. First, we will show that from
\eqref{eq:1.10} using the quantum canonicity of the coordinate system we get the
following relations:
\begin{subequations}
\label{eq:1.22}
\begin{align}
S(x^\alpha \star_M^{(x)} f) & = x^\alpha \star^{(x)} S(f), \label{eq:1.22a} \\
S(f \star_M^{(x)} x^\alpha) & = S(f) \star^{(x)} x^\alpha, \label{eq:1.22b}
\end{align}
\end{subequations}
for $f \in \Pol$. Indeed, we calculate that
\begin{equation}
x^\alpha \star_M^{(x)} f = x^\alpha f + \frac{1}{2}i\hbar
    \mathcal{P}^{\alpha\beta} \partial_{x^\beta}f,
\end{equation}
from which we get that
\begin{equation}
x^\alpha \star_M^{(x)} x^{\alpha_1} \dotsm x^{\alpha_r} =
    x^\alpha x^{\alpha_1} \dotsm x^{\alpha_r} + \frac{1}{2}i\hbar \sum_{s=1}^r
    \mathcal{P}^{\alpha\alpha_s} x^{\alpha_1} \dotsm x^{\alpha_{s-1}}
    x^{\alpha_{s+1}} \dotsm x^{\alpha_r}.
\label{eq:1.20}
\end{equation}
On the other hand from \eqref{eq:1.10} we have that
\begin{equation}
S(x^\alpha x^{\alpha_1} \dotsm x^{\alpha_r}) = \frac{1}{(r+1)!}
    \sum_{\sigma \in \mathfrak{S}_{r+1}} x^{\sigma(\alpha)} \star^{(x)}
    x^{\sigma(\alpha_1)} \star^{(x)} \dotsm \star^{(x)} x^{\sigma(\alpha_r)}.
\label{eq:1.19}
\end{equation}
After the commutation of $x^\alpha$ to the left in each term of the right hand
side of \eqref{eq:1.19}, and by using the quantum canonicity condition
\eqref{eq:1.7} we get
\begin{align}
S(x^\alpha x^{\alpha_1} \dotsm x^{\alpha_r}) & = \frac{1}{r!}
    \sum_{\sigma \in \mathfrak{S}_r} x^\alpha \star^{(x)} x^{\sigma(\alpha_1)}
    \star^{(x)} \dotsm \star^{(x)} x^{\sigma(\alpha_r)} \nonumber \\
& \quad {} - \frac{1}{2}i\hbar \frac{1}{(r-1)!} \sum_{\sigma \in \mathfrak{S}_r}
    \mathcal{P}^{\alpha \sigma(\alpha_1)} x^{\sigma(\alpha_2)} \star^{(x)}
    \dotsm \star^{(x)} x^{\sigma(\alpha_r)} \nonumber \\
& = x^\alpha \star^{(x)} S(x^{\alpha_1} \dotsm x^{\alpha_r})
    - \frac{1}{2}i\hbar \sum_{s=1}^r
    \mathcal{P}^{\alpha\alpha_s} S(x^{\alpha_1} \dotsm x^{\alpha_{s-1}}
    x^{\alpha_{s+1}} \dotsm x^{\alpha_r}).
\label{eq:1.21}
\end{align}
Combining \eqref{eq:1.20} and \eqref{eq:1.21} we receive
\begin{equation}
S(x^\alpha \star_M^{(x)} x^{\alpha_1} \dotsm x^{\alpha_r}) =
    x^\alpha \star^{(x)} S(x^{\alpha_1} \dotsm x^{\alpha_r}),
\end{equation}
which shows \eqref{eq:1.22a}. Equation \eqref{eq:1.22b} can be proved
analogically.

Adding \eqref{eq:1.22a} to \eqref{eq:1.22b} we get the following recurrence
relations on $S$:
\begin{equation}
S(x^\alpha f) = \frac{1}{2}\left(x^\alpha \star^{(x)} S(f)
    + S(f) \star^{(x)} x^\alpha\right), \quad f \in \Pol.
\label{eq:1.11}
\end{equation}
After expanding the \nbr{\star^{(x)}}product and morphism $S$ in the formula
\eqref{eq:1.11} we get
\begin{align}
\sum_{k=0}^\infty \hbar^k S_k(x^\alpha f) & = \frac{1}{2} \sum_{l=0}^\infty
    \sum_{n=0}^\infty \hbar^{l+n} \bigl(C_l(x^\alpha,S_n(f))
    + C_l(S_n(f),x^\alpha)\bigr) \nonumber \\
& = \sum_{k=0}^\infty \hbar^k \frac{1}{2} \sum_{l=0}^k
    \bigl(C_l(x^\alpha,S_{k-l}(f)) + C_l(S_{k-l}(f),x^\alpha)\bigr).
\label{eq:1.15}
\end{align}
From \eqref{eq:1.15} we get the following recurrence relations on $S_k$ for
$k \ge 0$:
\begin{equation}
S_k(x^\alpha f) = \frac{1}{2} \sum_{l=0}^k
    \bigl(C_l(x^\alpha,S_{k-l}(f)) + C_l(S_{k-l}(f),x^\alpha)\bigr),
\end{equation}
which can be rewritten in the form
\begin{equation}
[S_k,x^\alpha](f) = \frac{1}{2} \sum_{l=1}^k \bigl(C_l(x^\alpha,S_{k-l}(f))
    + C_l(S_{k-l}(f),x^\alpha)\bigr), \quad k \ge 1.
\label{eq:1.16}
\end{equation}
Before going further, we need a lemma.

\begin{lemma}[\citet{Dito:1999}]
\label{lem:1.1}
Let $\psi \colon \Pol \to C^\infty(U)$, where $U$ is an open subset of
$\mathbb{R}^d$, be an \nbr{\mathbb{C}}linear map such that $\psi(1) =
\psi(x^\alpha) = 0$, and let $\phi \colon C^\infty(U) \times C^\infty(U) \to
C^\infty(U)$ be a bidifferential operator vanishing on constants. If $\psi$
satisfies
\begin{equation}
[\psi,x^\alpha](f) = \phi(x^\alpha,f), \quad f \in \Pol,
\label{eq:1.17}
\end{equation}
then there exists exactly one differential operator $\eta$ on $U$ such that
$\psi = \eta|_\Pol$.
\end{lemma}

The term of order 1 in \eqref{eq:1.16} yields $[S_1,x^\alpha](f) =
\frac{1}{2}\bigl(C_1(x^\alpha,f) + C_1(f,x^\alpha)\bigr)$. The right hand side
of this equality is a bidifferential operator acting on $x^\alpha,f$ and
vanishing on constants. Hence, by virtue of Lemma~\ref{lem:1.1}, $S_1$ uniquely
extends to a differential operator. Now, through similar arguments and by
induction on $k$ each $S_k$ uniquely extends to a differential operator.
Clearly, the map $S$ can be naturally extended to a
\nbr{\mathbb{C}\lshad\hbar\rshad}linear map on $C^\infty(U)\lshad\hbar\rshad$.

Each monomial $x^{\alpha_1} \dotsm x^{\alpha_r}$ can be written as a
\nbr{\star_M^{(x)}}polynomial:
\begin{equation}
x^{\alpha_1} \dotsm x^{\alpha_r} = \frac{1}{r!} \sum_{\sigma \in \mathfrak{S}_r}
    x^{\sigma(\alpha_1)} \star_M^{(x)} \dotsm \star_M^{(x)}
    x^{\sigma(\alpha_r)}.
\end{equation}
Thus, using \eqref{eq:1.22a} we get for $f \in \Pol$
\begin{align}
S\bigl((x^{\alpha_1} \dotsm x^{\alpha_r}) \star_M^{(x)} f\bigr) & =
    \frac{1}{r!} \sum_{\sigma \in \mathfrak{S}_r} S(x^{\sigma(\alpha_1)}
    \star_M^{(x)} \dotsm \star_M^{(x)} x^{\sigma(\alpha_r)} \star_M^{(x)} f)
    \nonumber \\
& = \frac{1}{r!} \sum_{\sigma \in \mathfrak{S}_r} x^{\sigma(\alpha_1)}
    \star^{(x)} S(x^{\sigma(\alpha_2)} \star_M^{(x)}
    \dotsm \star_M^{(x)} x^{\sigma(\alpha_r)} \star_M^{(x)} f) \nonumber \\
& = \frac{1}{r!} \sum_{\sigma \in \mathfrak{S}_r} x^{\sigma(\alpha_1)}
    \star^{(x)} \dotsm \star^{(x)} x^{\sigma(\alpha_r)} \star^{(x)} S(f)
= S(x^{\alpha_1} \dotsm x^{\alpha_r}) \star^{(x)} S(f).
\end{align}
Hence, for polynomials $f,g$ we have
\begin{equation}
S(f \star_M^{(x)} g) = S(f) \star^{(x)} S(g).
\end{equation}
On the other hand $S$ can be used to define a star-product $\star'$ equivalent
to the \nbr{\star^{(x)}}product through the formula
\begin{equation}
S(f \star' g) = S(f) \star^{(x)} S(g).
\end{equation}
Both star-products $\star_M^{(x)}$ and $\star'$ agree on polynomials, therefore
must be equal, since two bidifferential operators equal on $\Pol \times \Pol$
are equal. This shows the existence of the searched morphism $S$. The uniqueness
follows from the fact that any morphism $S$ satisfying \eqref{eq:1.8} also
satisfies \eqref{eq:1.22}, and using quantum canonicity of the coordinate system
from this we deduce that $S$ on monomials takes the form \eqref{eq:1.10}.
\end{proof}

If the \nbr{\star}product besides the conditions
\itref{item:1.1}--\itref{item:1.5} satisfies also the parity condition
\begin{equation}
C_k(f,g) = (-1)^k C_k(g,f), \quad f,g \in C^\infty(M),
\end{equation}
then relations \eqref{eq:1.9} take the form
\begin{subequations}
\label{eq:1.23}
\begin{align}
[S_1,x^\alpha](f) & = 0,
\label{eq:1.23a} \\
[S_{2k+1},x^\alpha](f) & = \sum_{l=1}^k C_{2l}(x^\alpha,S_{2(k-l)+1}(f)),
\label{eq:1.23b} \\
[S_{2k},x^\alpha](f) & = \sum_{l=1}^k C_{2l}(x^\alpha,S_{2(k-l)}(f)),
\label{eq:1.23c}
\end{align}
\end{subequations}
for $k \ge 1$. From \eqref{eq:1.23a} follows that $S_1$ is an operator of
multiplication by function. By virtue of \eqref{eq:1.8b} this function has to be
equal 0. In the same manner from \eqref{eq:1.23b} $S_{2k+1} = 0$ for $k \ge 1$.
Thus, in this special case only terms of even order in the expansion of $S$ are
non-zero and they are given by \eqref{eq:1.23c}.

Note, that if $(x^1,\dotsc,x^d)$ is a purely quantum canonical coordinate
system, i.e.\ it is not at the same time classical canonical, then it must
depend on $\hbar$ and, in fact, will be a deformation of some classical
canonical coordinate system. The components $\mathcal{P}^{\mu\nu}$ of the
Poisson tensor $\mathcal{P}$ for such purely quantum canonical coordinate system
will also depend on $\hbar$ and can be expanded in the following series
\begin{equation}
\mathcal{P}^{\mu\nu} = \mathcal{P}^{\mu\nu}_0
    + \hbar \mathcal{P}^{\mu\nu}_1  + \hbar^2 \mathcal{P}^{\mu\nu}_2
    + o(\hbar^3),
\label{eq:1.24}
\end{equation}
where $\mathcal{P}^{\mu\nu}_0$ are of the form \eqref{eq:1.3}.
In consequence, the bidifferential operators $C_k$ from the expansion
\eqref{eq:1.1} of the \nbr{\star}product written in the coordinates
$(x^1,\dotsc,x^d)$ will depend on $\hbar$. Expanding $C_k$ in the power
series of $\hbar$ allows to write the \nbr{\star^{(x)}}product in the form
\begin{equation}
f \star^{(x)} g = \sum_{k=0}^\infty \hbar^k C'_k(f,g),
\end{equation}
where $C'_k$ are new bidifferential operators which are independent on $\hbar$,
and satisfy conditions \itref{item:1.1}--\itref{item:1.5}, where in condition
\itref{item:1.3} the Poisson bracket, in accordance to \eqref{eq:1.24}, is
associated to the Poisson tensor $\mathcal{P}_0$. As a result the
\nbr{\star^{(x)}}product can be considered as a coordinate representation, with
respect to the coordinate system $(x^1,\dotsc,x^d)$, of some star-product on a
Poisson manifold $(\mathcal{O},\mathcal{P}_0)$. The coordinates
$(x^1,\dotsc,x^d)$ are then classical and quantum canonical. Thus,
Theorem~\ref{thm:1.1} is also valid for a purely quantum canonical coordinate
system $(x^1,\dotsc,x^d)$. However, the Moyal product $\star_M^{(x)}$ will no
longer be associated to the Poisson tensor $\mathcal{P}$, but to some other
Poisson tensor.

\section{Form of the morphism $S$}
\label{sec:3}
In this section we will use \thmref{thm:1.1} to derive the form of the
morphism $S$. It is straightforward to calculate that the solution of
\eqref{eq:1.9}, in a general case, is of the form
\begin{equation}
S_k = \sum_{n=1}^\infty \frac{1}{n!}
    [x^{\alpha_1},\dotsc,[x^{\alpha_{n-1}},F_k^{\alpha_n}]]
    \partial_{x^{\alpha_1}} \dotsm \partial_{x^{\alpha_n}},
\label{eq:3.1}
\end{equation}
where $F_k^\alpha(f) = \frac{1}{2} \sum_{l=1}^k \bigl(C_l(x^\alpha,S_{k-l}(f))
+ C_l(S_{k-l}(f),x^\alpha)\bigr)$. Indeed,
\begin{align}
[S_k,x^\alpha] & = -\sum_{n=1}^\infty \frac{1}{n!}
    [x^\alpha,[x^{\beta_1},\dotsc,[x^{\beta_{n-1}},F_k^{\beta_n}]]]
    \partial_{\beta_1} \dotsm \partial_{\beta_n} \nonumber \\
& \quad {} + \sum_{n=1}^\infty\frac{1}{(n-1)!}
    [x^{\beta_1},\dotsc,[x^{\beta_{n-1}},F_k^\alpha]]
    \partial_{\beta_1} \dotsm \partial_{\beta_{n-1}} \nonumber \\
& = -\sum_{n=1}^\infty \frac{1}{n!}[x^{\beta_1},\dotsc,[x^{\beta_n},F_k^\alpha]]
    \partial_{\beta_1} \dotsm \partial_{\beta_n} \nonumber \\
& \quad {} + \sum_{n=0}^\infty \frac{1}{n!}
    [x^{\beta_1},\dotsc,[x^{\beta_n},F_k^\alpha]]
    \partial_{\beta_1} \dotsm \partial_{\beta_n} = F_k^\alpha.
\end{align}
Note, that when $C_k$ are bidifferential operators of finite order then the sum
in \eqref{eq:3.1} will be finite.

We will now calculate the formulas for the morphism $S$ for particular classes
of star-products.

\subsection{Case of a canonical star-product on $T^*\mathcal{Q}$ with a flat
base manifold $\mathcal{Q}$}
Let $\mathcal{Q}$ be an \nbr{n}dimensional manifold endowed with a flat
torsionless linear connection $\nabla$. Let us consider the cotangent bundle to
$\mathcal{Q}$, $M = T^*\mathcal{Q}$. On $M$ there exists a natural symplectic
form $\omega$, which induces a natural Poisson tensor $\mathcal{P} =
\omega^{-1}$. The linear connection $\nabla$ induces a flat torsionless
symplectic connection $\tilde{\nabla}$ on $M$, which Christoffel symbols in
induced canonical coordinates $(x^1,\dotsc,x^{2n}) =
(q^1,\dotsc,q^n,p_1,\dotsc,p_n)$ are given by the formula \cite{Plebanski:2001}
\begin{gather}
\tilde{\Gamma}^i_{jk} = \Gamma^i_{jk}, \quad
\tilde{\Gamma}^{\bar{i}}_{\bar{j} k} = -\Gamma^j_{ik}, \quad
\tilde{\Gamma}^{\bar{i}}_{j \bar{k}} = -\Gamma^k_{ji}, \quad
\tilde{\Gamma}^{\bar{i}}_{jk} = p_l(\Gamma^r_{jk} \Gamma^l_{ri}
    + \Gamma^r_{ik} \Gamma^l_{rj} - \Gamma^l_{ij,k}),
\label{eq:3.23}
\end{gather}
with the remaining components equal zero, where $\bar{i} = n + i$ and $,k$
denotes the partial derivative with respect to $q^k$. On
$(M,\omega,\tilde{\nabla})$ there exists a natural star-product given by the
formula \cite{Bayen:1978a,Bayen:1978b}
\begin{equation}
f \star g = \sum_{k=0}^\infty \frac{1}{k!} \left(\frac{i\hbar}{2}\right)^k
    \mathcal{P}^{\mu_1 \nu_1} \dotsm \mathcal{P}^{\mu_k \nu_k}
    (\underbrace{\tilde{\nabla} \dotsm \tilde{\nabla}}_k f)_{\mu_1\dotsc \mu_k}
    (\underbrace{\tilde{\nabla} \dotsm \tilde{\nabla}}_k g)_{\nu_1\dotsc \nu_k}.
\label{eq:3.36}
\end{equation}

If $(q^1,\dotsc,q^n)$ are coordinates on $\mathcal{Q}$ whose image is an open
subset $U \subset \mathbb{R}^n$ and $(x^1,\dotsc,x^{2n}) =
(q^1,\dotsc,q^n,p_1,\dotsc,p_n)$ are induced classical canonical coordinates on
$T^*\mathcal{Q}$ with image $T^*U = U \times \mathbb{R}^n$, then these
coordinates are quantum canonical with respect to the \nbr{\star}product. The
\nbr{\star}product can be written in these coordinates resulting in a
star-product in the algebra $C^\infty(T^*U)\lshad\hbar\rshad$. In this algebra
the Moyal product can also be defined. In accordance to \secref{sec:1} these two
star-products are equivalent and in what follows the equivalence morphism $S$ to
the fourth order in $\hbar$ will be derived.

It can be calculated that the operators $C_k(x^\alpha,\sdot)$ take the form
\begin{subequations}
\label{eq:3.24}
\begin{align}
C_k(q^j,\sdot) & = \frac{1}{k!} \left(\frac{i}{2}\right)^k
    (\underbrace{\nabla \dotsm \nabla}_{k} q^j)_{j_1 \dotsc j_k}
    \partial_{p_{j_1}} \dotsm \partial_{p_{j_k}}, \\
C_{k+1}(p_j,\sdot) & = \frac{1}{(k+1)!} \left(\frac{i}{2}\right)^{k+1}
    \left( f^r_{j j_1 \dotsc j_{k+1}}
    - (k + 1) f^l_{j (j_1 \dotsc j_k} \Gamma^r_{j_{k+1}) l} \right)
    p_r \partial_{p_{j_1}} \dotsm \partial_{p_{j_{k+1}}} \nonumber \\
& \quad {} - \frac{1}{k!} \left(\frac{i}{2}\right)^{k+1} f^l_{j j_1 \dotsc j_k}
    \partial_{q^l} \partial_{p_{j_1}} \dotsm \partial_{p_{j_k}}
    - \frac{1}{(k - 1)!} \left(\frac{i}{2}\right)^{k+1}
    f^r_{j l (j_1 \dotsc j_{k-1}} \Gamma^l_{j_k) r}
    \partial_{p_{j_1}} \dotsm \partial_{p_{j_k}},
\end{align}
\end{subequations}
where functions $f^l_{j j_1 \dotsc j_k}$ are given recursively by
\begin{subequations}
\begin{align}
f^l_{j j_1 \dotsc j_{k+1}} & = f^l_{j (j_1 \dotsc j_k,j_{k+1})}
    + f^r_{j (j_1 \dotsc j_k} \Gamma^l_{j_{k+1}) r}
    - k f^l_{j r (j_1 \dotsc j_{k-1}} \Gamma^r_{j_k j_{k+1})}, \\
f^l_{j j_1} & = \Gamma^l_{j j_1}.
\end{align}
\end{subequations}
We are using here the following notation: round brackets $()$ enclosing a group
of indices are to be understood as a symmetrization with respect to this group
of indices. Indeed, using \eqref{eq:3.23} one receives that
\begin{subequations}
\label{eq:3.25}
\begin{align}
(\underbrace{\tilde{\nabla} \dotsm \tilde{\nabla}}_k q^j)_{j_1 \dotsc j_k} & =
    (\underbrace{\nabla \dotsm \nabla}_k q^j)_{j_1 \dotsc j_k}, \\
(\underbrace{\tilde{\nabla} \dotsm \tilde{\nabla}}_{k+1} p_j)_{j_1 \dotsc
    j_{k+1}} & = \left(f^r_{j j_1 \dotsc j_{k+1}}
    - (k + 1) f^l_{j (j_1 \dotsc j_k} \Gamma^r_{j_{k+1}) l} \right) p_r, \\
(\underbrace{\tilde{\nabla} \dotsm \tilde{\nabla}}_{k+1} p_j)_{\bar{l} j_1
    \dotsc j_k} & = f^l_{j j_1 \dotsc j_k},
\end{align}
\end{subequations}
where remaining terms are equal zero, and
\begin{subequations}
\label{eq:3.26}
\begin{align}
(\underbrace{\tilde{\nabla} \dotsm \tilde{\nabla}}_k g)_{\bar{j}_1 \dotsc
    \bar{j}_k} & = \partial_{p_{j_1}} \dotsm \partial_{p_{j_k}} g, \\
(\underbrace{\tilde{\nabla} \dotsm \tilde{\nabla}}_{k+1} g)_{l \bar{j}_1 \dotsc
    \bar{j}_k} & = \partial_{q^l} \partial_{p_{j_1}} \dotsm \partial_{p_{j_k}} g
    + \Gamma^{j_1}_{lj} \partial_{p_j} \dotsm \partial_{p_{j_k}} g + \dotsb
    + \Gamma^{j_k}_{lj} \partial_{p_{j_1}} \dotsm \partial_{p_j} g.
\end{align}
\end{subequations}
From \eqref{eq:3.25} and \eqref{eq:3.26} one receives \eqref{eq:3.24}.

Using a computer algebra program one can calculate that the second and fourth
order terms in the expansion of the morphism $S$ with respect to $\hbar$ take
the form
\begin{subequations}
\label{eq:3.37}
\begin{align}
S_2 & = \frac{1}{8}\Gamma^i_{jk} \partial_{q^i}\partial_{p_j}\partial_{p_k}
    + \frac{1}{8}\Gamma^i_{lj}\Gamma^l_{ik} \partial_{p_j}\partial_{p_k}
    + \frac{1}{24} \left( 2\Gamma^i_{nl}\Gamma^n_{jk}
    - \Gamma^i_{jk,l} \right) p_i \partial_{p_j}\partial_{p_k}\partial_{p_l}, \\
S_4 & = S_{j_1 j_2 j_3 j_4} \partial_{p_{j_1}} \partial_{p_{j_2}}
    \partial_{p_{j_3}} \partial_{p_{j_4}}
    + S^i_{j_1 j_2 j_3 j_4} \partial_{q^i} \partial_{p_{j_1}}
    \partial_{p_{j_2}} \partial_{p_{j_3}} \partial_{p_{j_4}} \nonumber \\
& \quad {} + S^{i_1 i_2}_{j_1 j_2 j_3 j_4} \partial_{q^{i_1}}
    \partial_{q^{i_2}} \partial_{p_{j_1}} \partial_{p_{j_2}} \partial_{p_{j_3}}
    \partial_{p_{j_4}}
    + S^r_{j_1 j_2 j_3 j_4 j_5} p_r \partial_{p_{j_1}}\partial_{p_{j_2}}
    \partial_{p_{j_3}} \partial_{p_{j_4}} \partial_{p_{j_5}} \nonumber \\
& \quad {} + S^{ri}_{j_1 j_2 j_3 j_4 j_5} p_r \partial_{q^i}
    \partial_{p_{j_1}} \partial_{p_{j_2}} \partial_{p_{j_3}} \partial_{p_{j_4}}
    \partial_{p_{j_5}}
    + S^{rs}_{j_1 j_2 j_3 j_4 j_5 j_6} p_r p_s \partial_{p_{j_1}}
    \partial_{p_{j_2}} \partial_{p_{j_3}} \partial_{p_{j_4}} \partial_{p_{j_5}}
    \partial_{p_{j_6}},
\end{align}
\end{subequations}
where
\begin{subequations}
\begin{align}
S_{j_1 j_2 j_3 j_4} & = \frac{1}{384} \Bigl(
    - 4\Gamma^k_{l (j_1}\Gamma^l_{j_2 j_3,j_4) k}
    - 4\Gamma^k_{ln}\Gamma^l_{k (j_1}\Gamma^n_{j_2 j_3,j_4)}
    + 8\Gamma^k_{n (j_1}\Gamma^l_{|k| j_2}\Gamma^n_{j_3 j_4),l}
    + 8\Gamma^n_{kl,(j_1}\Gamma^k_{j_2 j_3}\Gamma^l_{j_4)n} \nonumber \\
& \quad {} + 8\Gamma^k_{ln}\Gamma^l_{m (j_1}\Gamma^n_{j_2 j_3}\Gamma^m_{j_4)k}
    + 3\Gamma^k_{l (j_1}\Gamma^l_{j_2 |k}\Gamma^n_{m| j_3}\Gamma^m_{j_4)n}
    - 6\Gamma^k_{m (j_1}\Gamma^l_{j_2 |k}\Gamma^n_{l| j_3}\Gamma^m_{j_4)n}
    \Bigr), \displaybreak[0] \\
S^i_{j_1 j_2 j_3 j_4} & = \frac{1}{384} \Bigl(
    - \Gamma^i_{(j_1 j_2,j_3 j_4)} + 5\Gamma^k_{(j_1 j_2}\Gamma^i_{j_3 j_4),k}
    - 2\Gamma^i_{k (j_1}\Gamma^k_{j_2 j_3,j_4)}
    + 6\Gamma^k_{l (j_1}\Gamma^i_{j_2 j_3}\Gamma^l_{j_4)k}
    + 2\Gamma^k_{(j_1 j_2}\Gamma^l_{j_3 j_4)}\Gamma^i_{kl}
    \Bigr), \displaybreak[0] \\
S^{i_1 i_2}_{j_1 j_2 j_3 j_4} & = \frac{1}{128}
    \Gamma^{i_1}_{(j_1 j_2}\Gamma^{i_2}_{j_3 j_4)}, \displaybreak[0] \\
S^r_{j_1 j_2 j_3 j_4 j_5} & = \frac{1}{1920} \Bigl(
    \Gamma^r_{(j_1 j_2,j_3 j_4 j_5)}
    - 9\Gamma^k_{(j_1 j_2}\Gamma^r_{j_3 j_4,j_5) k}
    - 2\Gamma^r_{k (j_1}\Gamma^k_{j_2 j_3,j_4 j_5)}
    + 3\Gamma^r_{k (j_1,j_2}\Gamma^k_{j_3 j_4,j_5)} \nonumber \\
& \quad {} - 7\Gamma^k_{l (j_1}\Gamma^l_{j_2 j_3,j_4}\Gamma^r_{j_5)k}
    - 9\Gamma^r_{kl}\Gamma^k_{(j_1 j_2}\Gamma^l_{j_3 j_4,j_5)}
    - 10\Gamma^k_{l (j_1}\Gamma^r_{j_2 j_3,j_4}\Gamma^l_{j_5)k}
    + 16\Gamma^r_{l (j_1}\Gamma^k_{j_2 j_3}\Gamma^l_{j_4 j_5),k}
    \nonumber \\
& \quad {} + 6\Gamma^r_{kl,(j_1}\Gamma^k_{j_2 j_3}\Gamma^l_{j_4 j_5)}
    + 20\Gamma^r_{k (j_1}\Gamma^k_{j_2 j_3}\Gamma^l_{|n| j_4}\Gamma^n_{j_5)l}
    + 16\Gamma^r_{n (j_1}\Gamma^k_{j_2 j_3}\Gamma^l_{j_4 j_5)}\Gamma^n_{kl}
    \Bigr), \displaybreak[0] \\
S^{ri}_{j_1 j_2 j_3 j_4 j_5} & = \frac{1}{192} \Bigl(
    - \Gamma^i_{(j_1 j_2}\Gamma^r_{j_3 j_4,j_5)}
    + 2\Gamma^r_{k (j_1}\Gamma^k_{j_2 j_3}\Gamma^i_{j_4 j_5)}
    \Bigr), \displaybreak[0] \\
S^{rs}_{j_1 j_2 j_3 j_4 j_5 j_6} & = \frac{1}{1152} \Bigl(
    \Gamma^r_{(j_1 j_2,j_3}\Gamma^s_{j_4 j_5,j_6)}
    - 4\Gamma^r_{k (j_1}\Gamma^k_{j_2 j_3}\Gamma^s_{j_4 j_5,j_6)}
    + 4\Gamma^r_{k (j_1}\Gamma^k_{j_2 j_3}\Gamma^l_{j_4 j_5}\Gamma^s_{j_6)l}
    \Bigr),
\end{align}
\end{subequations}
and, as before, round brackets enclosing a group of indices denote a
symmetrization with respect to this group of indices. The fixed indices (those
not used in the symmetrization) are distinguished by vertical lines.

Let us notice that the $S_2$ term is necessary for quantizations of Hamiltonians
which are quadratic and cubic in momenta \cite{Blaszak:2013b} while $S_4$ term
is important for quantizations of Hamiltonians which are of forth and fifth
order in momenta, respectively.

\subsection{Case of a star-product on a general symplectic manifold}
Let $(M,\omega,\tilde{\nabla})$ be a symplectic manifold endowed with a
symplectic torsionless linear connection $\tilde{\nabla}$. Consider on such
manifold a star-product which to the second order in $\hbar$ is of the form
\begin{multline}
f \star g = fg + \frac{i\hbar}{2} \mathcal{P}^{\mu\nu} (\tilde{\nabla}_\mu f)
    (\tilde{\nabla}_\nu g)
    + \frac{1}{2} \left(\frac{i\hbar}{2}\right)^2 \mathcal{P}^{\mu_1 \nu_1}
    \mathcal{P}^{\mu_2 \nu_2} \left(
    (\tilde{\nabla}\tilde{\nabla}f)_{\mu_1 \mu_2}
    (\tilde{\nabla}\tilde{\nabla}g)_{\nu_1 \nu_2} - a\tilde{R}_{\mu_1 \mu_2}
    (\tilde{\nabla}_{\nu_1} f)(\tilde{\nabla}_{\nu_2} g) \right) \\
    + o(\hbar^3),
\label{eq:3.31}
\end{multline}
where $\mathcal{P} = \omega^{-1}$, $a \in \mathbb{R}$ and $\tilde{R}_{\mu \nu}$
is a Ricci curvature tensor. An example of such product, in the case $a = 0$, is
a Fedosov star-product \cite{Fedosov:1994}. It can be checked that the
star-product \eqref{eq:3.31} is indeed associative up to the second order in
$\hbar$. Note, that in the case of a flat connection $\tilde{\nabla}$ this
star-product takes the form \eqref{eq:3.36}. Observe, moreover, that although
in the simplest case $a = 0$ the expansion of the star-products \eqref{eq:3.36}
and \eqref{eq:3.31} coincide up to the second order in $\hbar$, this will not
be the case for higher order terms as the star-product in the form
\eqref{eq:3.36} for a non-flat connection is not associative.

The second order term in the expansion of the morphism $S$ with respect to
$\hbar$ takes the form
\begin{equation}
S_2 = -\frac{1}{24} \tilde{\Gamma}_{\alpha \beta \gamma} \partial^\alpha
    \partial^\beta \partial^\gamma
    + \frac{1}{16} (\tilde{\Gamma}^\mu_{\nu\alpha}\tilde{\Gamma}^\nu_{\mu\beta}
    + a\tilde{R}_{\alpha\beta}) \partial^\alpha \partial^\beta,
\label{eq:3.33}
\end{equation}
where $\tilde{\Gamma}_{\alpha \beta \gamma} = \omega_{\alpha \delta}
\tilde{\Gamma}^\delta_{\beta \gamma}$ and $\partial^\alpha =
\mathcal{P}^{\alpha \beta} \partial_\beta$. To prove \eqref{eq:3.33} first note
that the condition that $\tilde{\nabla}$ has vanishing torsion can be restated
as
\begin{equation}
\tilde{\Gamma}^\alpha_{\beta \gamma} = \tilde{\Gamma}^\alpha_{\gamma \beta},
\label{eq:3.34}
\end{equation}
and the condition that $\tilde{\nabla}$ is symplectic ($\omega_{\mu \nu;\alpha}
= 0$, $\mathcal{P}^{\mu \nu}_{\phantom{\mu \nu};\alpha} = 0$) in canonical
coordinates can be restated as
\begin{subequations}
\label{eq:3.35}
\begin{align}
\mathcal{P}^{\delta \beta} \tilde{\Gamma}^\alpha_{\beta \gamma} & =
    \mathcal{P}^{\alpha \beta} \tilde{\Gamma}^\delta_{\beta \gamma},
\label{eq:3.35a} \\
\omega_{\delta \alpha} \tilde{\Gamma}^\alpha_{\beta \gamma} & =
    \omega_{\beta \alpha} \tilde{\Gamma}^\alpha_{\delta \gamma}.
\label{eq:3.35b}
\end{align}
\end{subequations}
From conditions \eqref{eq:3.34} and \eqref{eq:3.35b} we get that
$\tilde{\nabla}$ is symplectic and torsionless iff
$\tilde{\Gamma}_{\alpha \beta \gamma}$ is symmetric with respect to indices
$\alpha,\beta,\gamma$ \cite{Plebanski:2001}. Now, from \eqref{eq:3.31} and
\eqref{eq:3.35a} we get that
\begin{align}
C_2(x^\alpha,\sdot) & = -\frac{1}{8}\mathcal{P}^{\mu_1 \nu_1}
    \mathcal{P}^{\mu_2 \nu_2} \left(
    (\tilde{\nabla}\tilde{\nabla}x^\alpha)_{\mu_1 \mu_2}
    (\tilde{\nabla}\tilde{\nabla}(\sdot))_{\nu_1 \nu_2}
    - a\tilde{R}_{\mu_1 \mu_2}
    (\tilde{\nabla}_{\nu_1} x^\alpha)\tilde{\nabla}_{\nu_2} \right) \nonumber \\
& = \frac{1}{8}\mathcal{P}^{\mu_1 \nu_1} \mathcal{P}^{\mu_2 \nu_2} \left(
    \tilde{\Gamma}^\alpha_{\mu_1 \mu_2}(\partial_{\nu_1} \partial_{\nu_2}
    - \tilde{\Gamma}^\beta_{\nu_1 \nu_2} \partial_\beta)
    - a\tilde{R}_{\mu_1 \mu_2} \delta^\alpha_{\nu_1} \partial_{\nu_2} \right)
    \nonumber \\
& = \frac{1}{8} \tilde{\Gamma}^\alpha_{\mu_1 \mu_2} \partial^{\mu_1}
    \partial^{\mu_2}
    + \frac{1}{8}\mathcal{P}^{\mu_1 \alpha} \tilde{\Gamma}^{\nu_1}_{\mu_1 \mu_2}
    \tilde{\Gamma}^{\mu_2}_{\nu_1 \nu_2} \partial^{\nu_2}
    + \frac{1}{8}a\mathcal{P}^{\mu_1 \alpha} \tilde{R}_{\mu_1 \mu_2}
    \partial^{\mu_2}.
\end{align}
On the other hand
\begin{align}
[S_2,x^\alpha] & = -\frac{1}{24} \mathcal{P}^{\delta \alpha}
    \tilde{\Gamma}_{\delta\beta\gamma} \partial^\beta \partial^\gamma
    - \frac{1}{24} \mathcal{P}^{\beta \alpha}
    \tilde{\Gamma}_{\delta \beta \gamma} \partial^\delta \partial^\gamma
    - \frac{1}{24} \mathcal{P}^{\gamma \alpha}
    \tilde{\Gamma}_{\delta \beta \gamma} \partial^\delta \partial^\beta
    \nonumber \\
& \quad {} + \frac{1}{16} \mathcal{P}^{\gamma \alpha}
    \tilde{\Gamma}^\mu_{\nu \gamma}\tilde{\Gamma}^\nu_{\mu \beta} \partial^\beta
    + \frac{1}{16} \mathcal{P}^{\beta \alpha} \tilde{\Gamma}^\mu_{\nu \gamma}
    \tilde{\Gamma}^\nu_{\mu \beta} \partial^\gamma
    + \frac{1}{16}a \mathcal{P}^{\gamma \alpha} \tilde{R}_{\gamma\beta}
    \partial^\beta + \frac{1}{16}a \mathcal{P}^{\beta \alpha}
    \tilde{R}_{\gamma\beta} \partial^\gamma \nonumber \\
& = \frac{1}{8} \tilde{\Gamma}^\alpha_{\beta \gamma}\partial^{\beta}
    \partial^{\gamma}
    + \frac{1}{8} \mathcal{P}^{\gamma \alpha} \tilde{\Gamma}^{\mu}_{\nu \gamma}
    \tilde{\Gamma}^{\nu}_{\mu \beta} \partial^{\beta}
    + \frac{1}{8}a \mathcal{P}^{\gamma \alpha} \tilde{R}_{\gamma\beta}
    \partial^\beta.
\end{align}
Hence, $[S_2,x^\alpha] = C_2(x^\alpha,\sdot)$ which proves \eqref{eq:3.33}.


%

\end{document}